\documentclass[a4paper]{article}
\pdfoutput=1

\usepackage[T1]{fontenc}
\usepackage[utf8]{inputenc}
\usepackage{lmodern}
\usepackage{enumitem}
\usepackage[ruled,vlined]{algorithm2e}
\usepackage{amssymb,amsmath,amsthm}
\usepackage{graphicx}
\usepackage{hyperref}
\usepackage{url}
\usepackage{pgf,tikz}
\usetikzlibrary{arrows}
\usepackage{color}
\definecolor{ffqqqq}{rgb}{1,0,0}

\newcommand{\N}{\mathbb{N}}
\newcommand{\Z}{\mathbb{Z}}
\newcommand{\R}{\mathbb{R}}
\newcommand{\false}{\operatorname{false}}

\newcommand{\qealg}{\mathcal{A}}
\newcommand{\sgn}{\operatorname{sgn}}
\newcommand{\sgnl}{\operatorname{sgnl}}
\newcommand{\sgnr}{\operatorname{sgnr}}
\newcommand{\sgnseq}[1]{\overline{#1}}
\newcommand{\msub}[2]{{#1}\mathop{/\kern-3pt/}{#2}}
\newcommand{\Zym}{\Z[y_1,\dots,y_m]}
\newcommand{\evama}{\langle\alpha_1,\dots,\alpha_m\rangle}
\newcommand{\evamb}{\langle\beta_1,\dots,\beta_m\rangle}
\newcommand{\guard}[2]{\Gamma({#1},{#2})}

\theoremstyle{definition}
\newtheorem{lemma}{Lemma}

\theoremstyle{plain}

\newtheorem{theorem}[lemma]{Theorem}
\newtheorem{proposition}[lemma]{Proposition}

\begin{document}
\title{A Generalized Framework for Virtual Substitution}

\author{Marek Ko{\v s}ta and Thomas Sturm\\
Max-Planck-Institut f{\" u}r Informatik\\
Saarbr{\" u}cken, Germany\\
\url{{mkosta,sturm}@mpi-inf.mpg.de}}

\date{January 22, 2015}

\maketitle

\begin{abstract}
  We generalize the framework of virtual substitution for real quantifier
  elimination to arbitrary but bounded degrees. We make explicit the
  representation of test points in elimination sets using roots of parametric
  univariate polynomials described by Thom codes. Our approach follows an early
  suggestion by Weispfenning, which has never been carried out explicitly.
  Inspired by virtual substitution for linear formulas, we show how to
  systematically construct elimination sets containing only test points
  representing lower bounds.
\end{abstract}

\section{Introduction}\label{SE:intro}
After Tarski's seminal paper~\cite{Tarski:48a} there has been considerable
research on quantifier elimination for real closed fields. One research line
lead to complexity results and asymptotically fast
algorithms~\cite{BasuPollack:96a,Grigoriev:1988a,Renegar:1992b,Weispfenning:88a}.
Unfortunately, these algorithms turned out not to be feasible in
practice~\cite{Hong:91a}.

From a practical point of view, the invention of quantifier elimination by
cylindrical algebraic decomposition (CAD) was an important
step~\cite{Collins:75}. Several
enhancements~\cite{CollinsHong:91,McCallum:1984a} of the original procedure
combined with an efficient implementation~\cite{Brown:03a} made it possible to
apply CAD-based algorithms to real-world problems to some
extent~\cite{HongEtAl:1997}. One principal drawback of all CAD-based algorithms
is the fact that parameters significantly contribute to the theoretical
complexity~\cite{Brown:2007}.

Virtual substitution is an alternative approach, particularly strong for
formulas with low degrees of the quantified variables, which is not sensitive to
the number of parameters. Similarly to CAD there exist efficient
implementations~\cite{DolzmannSturm:1997a}. The original description of the
method as well as first improvements and implementations were limited to linear
formulas~\cite{Weispfenning:88a,LoosWeispfenning:1993a}. The next important step
focused on formulas up to total degree two of the quantified
variables~\cite{Weispfenning:1997a}. That publication furthermore suggested in a
very abstract way to extend the procedure to arbitrarily large degree bounds,
and mentioned Thom's lemma~\cite{Roy:1998} as a possibility for distinguishing
real polynomial roots. In another publication Weispfenning made precise how to
perform virtual substitution up to degree three, without using Thom codes,
however~\cite{Weispfenning:1994a}.

Recently, virtual substitution has been playing some role in satisfiability
modulo theory solving~\cite{AbrahamLoup:10a,CorziliusAbraham:11a}. There appears
to be considerable interest in higher degrees~\cite{Liiva:2014}, focusing on the
software aspect rather than on theoretical foundations.

The present work picks up the original idea of Thom codes for generalizing
virtual substitution to arbitrary but fixed degree bounds, developing
theoretically precise foundations and a rigorous framework. Our original
contributions are the following:
\begin{enumerate}
\item We describe an encoding of parametric polynomial roots based on Thom's
  lemma. We prove that this encoding uniquely determines a root of a parametric
  univariate polynomial.
\item For a given encoding we specify formal necessary and sufficient conditions
  for the existence of a corresponding root. This allows to discard redundant
  elimination terms.
\item Our encoding allows to easily identify roots representing lower in
  contrast to upper bounds of relevant intervals.
\item We generally reduce the size of elimination sets by considering
  exclusively lower bounds of relevant intervals. This improves even the
  well-known elimination sets for the quadratic case.
\end{enumerate}

The plan of the paper is as follows: In Section~\ref{SE:roots} we describe our
root encoding based on Thom's lemma. In Section~\ref{SE:framework} we introduce
our framework by specifying what elimination terms look like and how they are
substituted. We prove the correctness of the framework. In
Section~\ref{SE:smaller} we reduce the size of elimination sets by considering
exclusively lower bounds of relevant intervals. In Section~\ref{SE:relation} we
discuss our framework in the context of alternative approaches in the
literature. We furthermore point at its compatibility with various
generalizations of quantifier elimination by virtual substitution. In
Section~\ref{SE:examples} we discuss practical issues related to possible
implementation strategies.

\section{Parametric Roots}\label{SE:roots}
Let $p\in\Zym[x]$ with $\deg p = n$, and let $\alpha_1$, \dots,~$\alpha_m\in\R$.
Denote by $\evama: \Zym\to\R$ the evaluation homomorphism in postfix notation,
i.e., we have $p\evama\in\R[x]$. Consider arbitrary $f\in\R[x]$. For $\xi\in\R$
we define the \emph{sign sequence} of length $k$ at $\xi$ as follows:
\begin{displaymath}
  \sgn_{\xi}(f, k) = (\sgn(f(\xi)), \sgn(f'(\xi)), \dots,
  \sgn(f^{(k-1)}(\xi))).
\end{displaymath}
A sign sequence $\sgnseq{s}=(s_1,\dots,s_n)\in\{-1,0,1\}^n$ is \emph{consistent}
with $p$ if there exist $\alpha_1$, \dots,~$\alpha_m$, $\xi\in\R$ such that the
following conditions hold:
\begin{enumerate}[label=($\textrm{C}_\arabic*$),ref=($\textrm{C}_\arabic*$)]
\item $\deg p\evama > 0$,\label{C1}
\item $p\evama(\xi) = 0$,\label{C2}
\item $\sgn_{\xi}(p\evama', n)=\sgnseq{s}$.\label{C3}
\end{enumerate}
It follows that $(0,\dots,0)$ is not consistent with any $q\in\Zym[x]$. The idea
is that $\sgnseq{s}$ uniquely describes a real root $\xi$ of $p\evama$.

A \emph{guard} $\guard{p}{\sgnseq{s}}$ for $p$ and
$\sgnseq{s}=(s_1,\dots,s_n)\in\{-1,0,1\}^n$ is a quantifier-free equivalent in
variables $y_1$, \dots,~$y_m$ of the Tarski formula
\begin{displaymath}
  \exists x\Bigl(p = 0 \land
  \bigwedge_{i=1}^{|\sgnseq{s}|} (p^{(i)}\mathrel{\sigma(s_i)}0)\Bigr),
\end{displaymath}
where
\begin{displaymath}
  \sigma(s_i) =
  \begin{cases}
    \text{``$<$''} &\text{if $s_i=-1$},\\
    \text{``$=$''} &\text{if $s_i=0$},\\
    \text{``$>$''} &\text{if $s_i=1$}.
  \end{cases}
\end{displaymath}

\begin{lemma}\label{LE:consistent2}
  Let $p\in\Z[y_1,\dots,y_m][x]$ with $\deg p = n > 0$. Let
  $\sgnseq{s}\neq(0,\dots,0)$ be a sign sequence of length $n$, and let
  $\alpha_1$, \dots,~$\alpha_m\in\R$. The following are equivalent:
  \begin{enumerate}[label=(\roman*)]
  \item $\guard{p}{\sgnseq{s}}$ holds for $\alpha_1$, \dots,~$\alpha_m$.
  \item There exists a unique $\xi\in\R$ such that conditions \ref{C1}--\ref{C3}
    hold.
  \end{enumerate}
\end{lemma}
\begin{proof}
  Assume (i). It is easy to see that there exists $\xi\in\R$ such that \ref{C2}
  and \ref{C3} hold. Denote $p\evama$ by $p_{\alpha}$. Then $s\neq(0,\dots,0)$
  implies that at least one of polynomials $p_{\alpha}'$,
  \dots,~$p_{\alpha}^{(n)}$, is not identically zero. Hence, $p_{\alpha}$ is of
  positive degree, i.e., \ref{C1} holds as well. The uniqueness of $\xi$ follows
  directly from Thom's little lemma~\cite{Roy:1998}.

  Assume (ii). Then it is easy to see that $\alpha_1$, \dots,~$\alpha_m$ satisfy
  $\guard{p}{\sgnseq{s}}$, even with a unique $x$.
\end{proof}

Lemma~\ref{LE:consistent2} ensures that $\sgnseq{s}\neq (0,\dots,0)$ is
consistent with $p$ if and only if $\guard{p}{\sgnseq{s}}$ is satisfiable. In
the positive case, $\guard{p}{\sgnseq{s}}$ gives a necessary and sufficient
condition for $p$ to have a root that is described by $\sgnseq{s}$.

As an example consider $p = ax^2 + bx + c$. Table~\ref{TAB:ex1} lists all sign
sequences consistent with $p$ along with their respective guards. The
polynomials $f$ in the table are obtained from $p$ by substituting for $a$, $b$,
and $c$ suitable values satisfying the corresponding condition
$\guard{p}{\sgnseq{s}}$ in the table. They all have $1$ as a root described by
the corresponding $\sgnseq{s}$, formally, $\sgn_{1}(f',2)=\sgnseq{s}$.
\begin{table}
  \renewcommand{\tabcolsep}{1em}
  \centering
  \begin{tabular}{lll}
    \hline
    $\sgnseq{s}$ & $\guard{p}{\sgnseq{s}}$ & $f$\\\hline
    $(-1,-1)$  & $a < 0 \land 4 a c - b^{2} < 0$ & $-x^2 + x$\\
    $(-1,0)$  & $a = 0 \land b < 0$ & $-x + 1$\\
    $(-1,1)$  & $a > 0 \land 4 a c - b^{2} < 0$ & $x^2 - 3x + 2$\\
    $(0,-1)$  & $a < 0 \land 4 a c - b^{2} = 0$ & $-x^2 + 2x - 1$\\
    $(0,1)$  & $a > 0 \land 4 a c - b^{2} = 0$ & $x^2 - 2x + 1$\\
    $(1,-1)$  & $a < 0 \land 4 a c - b^{2} < 0$ & $-x^2 + 3x + 2$\\
    $(1,0)$  & $a = 0 \land b > 0$ & $x - 1$\\
    $(1,1)$  & $a > 0 \land 4 a c - b^{2} < 0$ & $x^2 - x$\\
    \hline
  \end{tabular}
  \caption{Consistent sign sequences and guards}
  \label{TAB:ex1}
\end{table}

Let $\sgnseq{s}$ be consistent with $p$. We call the pair $(p,\sgnseq{s})$ a
\emph{parametric root} of $p$. Recall from the definition of consistency of
$\sgnseq{s}$ with $p$ that there are particular $\alpha_1$,
\dots,~$\alpha_m\in\R$ such that $\sgnseq{s}$ uniquely describes a real root of
$p\evama$. Note that these $\alpha_1$, \dots,~$\alpha_m\in\R$ need not be
unique. The possible choices for $\alpha_1$, \dots,~$\alpha_m$ are in fact
described by $\guard{p}{\sgnseq{s}}$.

The following lemma considers two such possible choices, and states that the
signs of the obtained univariate polynomials in the neighborhoods of the
corresponding roots are invariant with respect to the two choices.

\begin{lemma}\label{LE:sgnclose}
  Let $p\in\Zym[x]$ with $\deg p = n > 0$. Let $\sgnseq{s}=(s_1,\dots,s_n)$ be
  consistent with $p$, where $\alpha_1$, \dots,~$\alpha_m$, $\xi\in\R$, and
  $\beta_1$, \dots,~$\beta_m$, $\zeta\in\R$ are two possible sets of choices in
  \ref{C1}--\ref{C3}. Denote $p\evama$ by $p_{\alpha}$, and $p\evamb$ by
  $p_{\beta}$. Let $\varepsilon$ be a positive infinitesimal number. Then the
  following hold:
  \begin{enumerate}[label=(\roman*)]
  \item $\sgn(p_{\alpha}(\xi-\varepsilon)) =
    \sgn(p_{\beta}(\zeta-\varepsilon))$
  \item $\sgn(p_{\alpha}(\xi+\varepsilon)) =
    \sgn(p_{\beta}(\zeta+\varepsilon))$.
  \end{enumerate}
\end{lemma}
\begin{proof}
  We show only (i), the proof of (ii) is similar. To start with, note that
  \ref{C1}--\ref{C3} ensure that
  \begin{equation}\label{EQ:hugo}
    \sgn_{\xi}(p_{\alpha}', n) = \sgn_{\zeta}(p_{\beta}', n) = \sgnseq{s}.
  \end{equation}
  We show by induction that for $i\in\{n,\dots,0\}$ the following holds:
  \begin{displaymath}
    \sgn(p_{\alpha}^{(i)}(\xi-\varepsilon)) =
    \sgn(p_{\beta}^{(i)}(\zeta-\varepsilon)).
  \end{displaymath}
  For $i=n$, observe that $p_{\alpha}^{(n)}$, $p_{\beta}^{(n)}\in\R$.
  From~(\ref{EQ:hugo}) it follows that $\sgn(p_{\alpha}^{(n)}(\xi)) =
  \sgn(p_{\beta}^{(n)}(\zeta)) = s_n$, and we can conclude that
  $\sgn(p_{\alpha}^{(n)}(\xi-\varepsilon)) =
  \sgn(p_{\beta}^{(n)}(\zeta-\varepsilon))$. Let now $k\in\{n-1,\dots,0\}$, and
  assume that $\sgn(p_{\alpha}^{(k+1)}(\xi-\varepsilon)) =
  \sgn(p_{\beta}^{(k+1)}(\zeta-\varepsilon))$. We have to show that
  $\sgn(p_{\alpha}^{(k)}(\xi-\varepsilon)) =
  \sgn(p_{\beta}^{(k)}(\zeta-\varepsilon))$. We distinguish cases. If $s_k\neq
  0$, then (\ref{EQ:hugo}) ensures that $\sgn(p_{\alpha}^{(k)}(\xi)) =
  \sgn(p_{\beta}^{(k)}(\zeta)) = s_k \neq 0$. If follows that
  $\sgn(p_{\alpha}^{(k)}(\xi-\varepsilon)) =
  \sgn(p_{\beta}^{(k)}(\zeta-\varepsilon))$, because $\varepsilon$ is
  infinitesimal. Assume now that $s_k=0$, and distinguish three cases:
  \begin{enumerate}[label=(\alph*)]
  \item $\sgn(p_{\alpha}^{(k)}(\xi-\varepsilon)) = 0$: Since
    $\sgn(p_{\alpha}^{(k)}(\xi)) = s_k = 0$, and $\varepsilon$ is infinitesimal,
    $p_{\alpha}^{(k)}$ is the zero polynomial. Thus, $p_{\alpha}^{(k+1)}$ is the
    zero polynomial as well. The induction hypothesis and our assumption $s_k=0$
    yield that $p_{\beta}^{(k+1)}$ is the zero polynomial as well. This means
    that $p_{\beta}^{(k)}$ is a constant polynomial. On the other hand,
    $\sgn(p_{\beta}^{(k)}(\zeta)) = s_k = 0$. Therefore, $p_{\beta}^{(k)}$ is
    the zero polynomial, in particular $\sgn(p_{\beta}^{(k)}(\zeta-\varepsilon))
    = 0$.
  \item $\sgn(p_{\alpha}^{(k)}(\xi-\varepsilon)) = 1$: Since $\varepsilon$ is
    positive and infinitesimal and $p_{\alpha}^{(k)}(\xi-\varepsilon) >
    p_{\alpha}^{(k)}(\xi)$, it follows that $p_{\alpha}^{(k)}$ is decreasing at
    $\xi-\varepsilon$. Therefore, we have
    $\sgn(p_{\alpha}^{(k+1)}(\xi-\varepsilon)) = -1$. By the induction
    hypothesis, it follows that $\sgn(p_{\beta}^{(k+1)}(\zeta-\varepsilon)) =
    -1$. Therefore, $p_{\beta}^{(k+1)}$ is decreasing at $\zeta-\varepsilon$.
    Finally, our assumption $s_k=0$ ensures that $\sgn(p_{\beta}^{(k)}(\zeta)) =
    s_k = 0$. Since $\varepsilon$ is infinitesimal, we have
    $p_{\beta}^{(k)}(\zeta-\varepsilon) > 0$, i.e.,
    $\sgn(p_{\beta}^{(k)}(\zeta-\varepsilon)) = 1$.
  \item $\sgn(p_{\alpha}^{(k)}(\xi-\varepsilon)) = -1$: Similar to (b).\qedhere
  \end{enumerate}
\end{proof}

Consider a parametric root $(p,\sgnseq{s})$, and let $\alpha_1$, \dots,
$\alpha_m\in\R$ be such that $\guard{p}{\sgnseq{s}}$ holds. By
Lemma~\ref{LE:consistent2} there exists a unique $\xi\in\R$ such that
\ref{C1}--\ref{C3} hold. We define the \emph{left sign} of $(p,\sgnseq{s})$ as
$\sgnl(p,\sgnseq{s}) = \sgn(p\evama(\xi-\varepsilon))$. Similarly, the
\emph{right sign} of $(p,\sgnseq{s})$ is defined as $\sgnr(p,\sgnseq{s}) =
\sgn(p\evama(\xi+\varepsilon))$. This is well-defined by
Lemma~\ref{LE:sgnclose}. Notice that~\ref{C1} ensures that both
$\sgnl(p,\sgnseq{s})$ and $\sgnr(p,\sgnseq{s})$ cannot be zero.

The left and the right sign of $(p,\sgnseq{s})$ can be computed as follows:
\begin{enumerate}
\item Find $\alpha_1$, \dots,~$\alpha_m\in\R$ satisfying
  $\guard{p}{\sgnseq{s}}$.
\item Compute an isolating interval $\left]l,r\right[$ of the root $\xi$ of
  $p\evama$ identified by $\sgnseq{s}$ such that $p\evama(l)\neq 0$ and
  $p\evama(r)\neq 0$.
\end{enumerate}
Then $\sgnl(p,\sgnseq{s}) = \sgn(p\evama(l))$ and $\sgnr(p,\sgnseq{s}) =
\sgn(p\evama(r))$.

\section{Elimination Terms and Elimination Sets}\label{SE:framework}
Let $n\in\N\setminus\{0\}$, and let $\varphi$ be an $\land$-$\lor$-combination
of atomic formulas $\{p_i\mathrel{\varrho_i}0\}_{i\in I}$, where $p_i\in\Zym[x]$
and $\varrho_i\in\{=,\neq,<,\leqslant,\geqslant,>\}$. Assume that $\deg
p_i\leqslant n$ for all $i\in I$. We say that $\varphi$ is of degree at most $n$
in $x$. In this section we are going to describe a method for eliminating
$\exists x$ from the formula $\exists x\varphi$.

Our method is not self-contained. It depends on an algorithm $\qealg$ that is
capable of eliminating a single existential quantifier from formulas of degree
$n$ in $x$, which have a very particular shape. We are going to describe these
formulas along with the definition of virtual substitution of elimination terms.
Later in Section~\ref{SE:examples} we will show that it is even sufficient to
consider only finitely many such formulas. Since the real numbers admit
effective quantifier elimination, it is clear that such an algorithm $\qealg$
exists. The key challenge with the approach proposed here is going to be to find
\emph{short} quantifier-free equivalents, possibly including considerable human
intelligence at least for post-processing.

We start with the description of the set $E_i$ of \emph{elimination terms}
generated by a particular atomic formula $(p_i\mathrel{\varrho_i}0)$ in $\varphi$:
\begin{displaymath}
  E_i =
  \begin{cases}
    \{(p_i,\sgnseq{s}) \mid \text{$(p_i,\sgnseq{s})$
      is a parametric root of $p_i$}\}
    &\text{if $\varrho_i\in\{=,\leqslant,\geqslant\}$},\\
    \{(p_i,\sgnseq{s})+\varepsilon \mid \text{$(p_i,\sgnseq{s})$
      is a parametric root of $p_i$}\}
    &\text{if $\varrho_i\in\{\neq,<,>\}$.}
  \end{cases}
\end{displaymath}

Consider an atomic formula $(q\mathrel{\varrho}0)$, where
$q=b_nx^n+\dots+b_1x+b_0\in\Zym[x]$. The \emph{virtual substitution}
$(q\mathrel{\varrho}0)[\msub{x}{r}]$ of an elimination term $r=(p,\sgnseq{s})$
for $x$ into $(q\mathrel{\varrho}0)$ is the quantifier-free formula computed by
$\qealg$ for
\begin{displaymath}
  \exists x\Bigl(
  p = 0 \land
  \bigwedge_{i=1}^{|\sgnseq{s}|} (p^{(i)}\mathrel{\sigma(s_i)}0) \land
  (q\mathrel{\varrho}0)
  \Bigr).
\end{displaymath}

Let $N$ and $P$ be the sets of all sign sequences $\sgnseq{t}$ consistent with
$q$ such that $\sgnr(q,\sgnseq{t}) < 0$ and $\sgnr(q,\sgnseq{t}) > 0$,
respectively. Furthermore, let be $\nu((p, \sgnseq{s}), (q, \sgnseq{t}))$ the
quantifier-free formula computed by $\qealg$ for
\begin{displaymath}
  \exists x\Bigl(
  p = 0 \land
  \bigwedge_{i=1}^{|\sgnseq{s}|} (p^{(i)}\mathrel{\sigma(s_i)}0) \land
  q = 0 \land
  \bigwedge_{i=1}^{|\sgnseq{t}|} (q^{(i)}\mathrel{\sigma(t_i)}0)
  \Bigr).
\end{displaymath}
In these terms, the virtual substitution
$(q\mathrel{\varrho}0)[\msub{x}{r+\varepsilon}]$ of $r+\varepsilon$ for $x$ is
defined as follows:
\begin{align*}
  (q = 0)[\msub{x}{r+\varepsilon}] &:\quad
  b_n = 0 \land \dots \land b_0 = 0\\
  (q \neq 0)[\msub{x}{r+\varepsilon}] &:\quad
  b_n \neq 0\lor \dots \lor b_0 \neq 0\\
  (q < 0)[\msub{x}{r+\varepsilon}] &:\quad (q < 0)[\msub{x}{r}] \lor
  \bigvee_{\sgnseq{t}\in N}\nu((p, \sgnseq{s}), (q, \sgnseq{t}))\\
  (q > 0)[\msub{x}{r+\varepsilon}] &:\quad (q > 0)[\msub{x}{r}] \lor
  \bigvee_{\sgnseq{t}\in P}\nu((p, \sgnseq{s}), (q, \sgnseq{t}))\\
  (q \leqslant 0)[\msub{x}{r+\varepsilon}] &:\quad (q < 0)[\msub{x}{r+\varepsilon}] \lor
  (q = 0)[\msub{x}{r+\varepsilon}]\\
  (q \geqslant 0)[\msub{x}{r+\varepsilon}] &:\quad (q > 0)[\msub{x}{r+\varepsilon}] \lor (q
  = 0)[\msub{x}{r+\varepsilon}].
\end{align*}

Note that, in contrast to~\cite[Section 3]{Weispfenning:1997a}, our definition
of virtual substitution of $r+\varepsilon$ is not recursive. The deeply nested
subformulas introduced with the recursive definition of $\nu$ for quadratic
quantifier elimination in~\cite[Section 3]{Weispfenning:1997a} are a
considerable obstacle for simplification. A good choice of $\qealg$ might help
to overcome this.

Another somewhat special elimination term is $-\infty$. It is not generated by
any atomic formula, but will generally occur in every elimination set. The
virtual substitution $(q\mathrel{\varrho} 0)[\msub{x}{-\infty}]$ of $-\infty$
for $x$ is defined as follows:
\begin{align*}
  (q = 0)[\msub{x}{-\infty}]&:\quad b_n = 0\land\dots\land b_0 = 0\\
  (q \neq 0)[\msub{x}{-\infty}]&:\quad b_n \neq 0\lor\dots\lor b_0 \neq 0\\
  (q < 0)[\msub{x}{-\infty}]&:\quad \mu_{<}(q,n)\lor\dots\lor\mu_{<}(q,0)\\
  (q > 0)[\msub{x}{-\infty}]&:\quad \mu_{>}(q,n)\lor\dots\lor\mu_{>}(q,0)\\
  (q \leqslant 0)[\msub{x}{-\infty}]&:\quad (q < 0)[\msub{x}{-\infty}] \lor (q =
  0)[\msub{x}{-\infty}]\\
  (q \geqslant 0)[\msub{x}{-\infty}]&:\quad (q > 0)[\msub{x}{-\infty}] \lor (q =
  0)[\msub{x}{-\infty}],
\end{align*}
where for $k\in\{0,\dots,n\}$ we define
\begin{align*}
  \mu_{<}(q, k)&:\quad b_n = 0 \land \dots \land b_{k+1} = 0 \land (-1)^k b_k < 0\\
  \mu_{>}(q, k)&:\quad b_n = 0 \land \dots \land b_{k+1} = 0 \land (-1)^k b_k > 0.
\end{align*}

As usual, the key idea of virtual substitution is to map not terms to terms, but
atomic formulas to quantifier-free formulas, which gives us the freedom to
describe its results in the Tarski language. That mapping on atomic formulas
naturally induces virtual substitution on arbitrary quantifier-free formulas.

In the sequel we prove that the terms generated by all atomic formulas occurring
in $\varphi$ together with $-\infty$ constitute an \emph{elimination set} $E$ for
$\exists x\varphi$ in the following sense:
\begin{displaymath}
  \R\models\exists x\varphi
  \longleftrightarrow
  \bigvee_{e\in E}\varphi[\msub{x}{e}].
\end{displaymath}

The following lemma is an immediate consequence of our definition of virtual
substitution:
\begin{lemma}\label{LE:implguard}
  Let $(q\mathrel{\varrho}0)$ be an atomic formula, where $q\in\Zym[x]$.
  Consider a parametric root $(p,\sgnseq{s})$. Assign values $\alpha_1$,
  \dots,~$\alpha_m\in\R$ to the variables $y_1$, \dots,~$y_m$. Then the
  following hold:
  \begin{enumerate}[label=(\roman*)]
  \item If $(q\mathrel{\varrho}0)[\msub{x}{(p,\sgnseq{s})}]$ holds, then
    $\guard{p}{\sgnseq{s}}$ holds as well.
  \item If $(q\mathrel{\varrho}0)[\msub{x}{(p,\sgnseq{s})+\varepsilon}]$ holds,
    and $\varrho\in\{<,>\}$, then $\guard{p}{\sgnseq{s}}$ holds as well.\qed
  \end{enumerate}
\end{lemma}

\begin{lemma}\label{LE:semantics1}
  Consider $\{p_k\mathrel{\varrho_k}0\}_{k\in K}$, where $K$ is finite,
  $p_k\in\Zym[x]$, and $\varrho_k\in\{=,\neq,<,\leqslant,\geqslant,>\}$. Assign
  values $\alpha_1$, \dots,~$\alpha_m\in\R$ to the variables $y_1$,
  \dots,~$y_m$. Consider a parametric root $(p,\sgnseq{s})$.
  \begin{enumerate}[label=(\roman*)]
  \item There exists $\vartheta\in\R$ such that for all $k\in K$ the following
    holds: If $(p_k\mathrel{\varrho_k}0)[\msub{x}{(p,\sgnseq{s})}]$ holds, then
    $\vartheta$ satisfies $(p_k\mathrel{\varrho_k}0)$.
  \item There exists $\zeta\in\R$ such that for all $k\in K$ the following
    holds: If $(p_k\mathrel{\varrho_k}0)[\msub{x}{(p,\sgnseq{s})+\varepsilon}]$
    holds, then $\zeta$ satisfies $(p_k\mathrel{\varrho_k}0)$.
  \end{enumerate}
\end{lemma}
\begin{proof}
  To begin the proof of (i), note that if there is no $k\in K$ such that
  $(p_k\mathrel{\varrho_k}0)[\msub{x}{(p,\sgnseq{s})}]$ holds, there is nothing
  to prove. Assume w.l.o.g. that
  $(p_1\mathrel{\varrho_1}0)[\msub{x}{(p,\sgnseq{s})}]$ holds.
  Lemma~\ref{LE:implguard} implies that $\guard{p}{\sgnseq{s}}$ holds as well.
  Consequently, Lemma~\ref{LE:consistent2} ensures that there is a unique
  $\xi\in\R$ such that \ref{C1}--\ref{C3} hold. Let now $k\in K$ and assume that
  $(p_k\mathrel{\varrho_k}0)[\msub{x}{(p,\sgnseq{s})}]$, i.e.,
  \begin{displaymath}
    \exists x\Bigl(
    p = 0 \land
    \bigwedge_{i=1}^{|\sgnseq{s}|} (p^{(i)}\mathrel{\sigma(s_i)}0) \land
    (p_k\mathrel{\varrho_k}0)
    \Bigr),
  \end{displaymath}
  holds. The uniqueness of $\xi$ satisfying conditions \ref{C1}--\ref{C3} now
  ensures that $(p_k\mathrel{\varrho_k}0)$ holds for $\xi$. This proves (i).

  To prove (ii) we first assume that there is some $l\in K$ such that
  $(p_l<0)[\msub{x}{(p,\sgnseq{s})+\varepsilon}]$ or
  $(p_l>0)[\msub{x}{(p,\sgnseq{s})+\varepsilon}]$ holds. Again,
  Lemma~\ref{LE:implguard} ensures that $\guard{p}{\sgnseq{s}}$ holds, and
  Lemma~\ref{LE:consistent2} ensures that there is a unique $\xi\in\R$ such that
  conditions \ref{C1}--\ref{C3} hold. Pick $\zeta$ from the open interval
  $]\xi,\xi'[$, where
  \begin{displaymath}
    \xi' = \min\{\delta\mid j\in K\land\deg p_j>0 \land p_j(\delta) = 0 \land
    \delta > \xi\}.
  \end{displaymath}
  Let now $k\in K$, and assume that
  $(p_k\mathrel{\varrho_k}0)[\msub{x}{(p,\sgnseq{s})+\varepsilon}]$ holds. If
  $\varrho_k$ is ``$=$,'' then $p_k$ is identically zero, so $(p_k=0)$ trivially
  holds for $\zeta$. If $\varrho_k$ is ``$<$,'' then there are two cases to
  consider:
  \begin{enumerate}[label=(\alph*)]
  \item In the first case $(p_k<0)[\msub{x}{(p,\sgnseq{s})}]$ holds. Since
    $\xi\in\R$ is unique satisfying conditions \ref{C1}--\ref{C3}, we deduce
    that $(p_k<0)$ holds at $\xi$. The choice of $\zeta$ now guarantees that
    there is no root of $p_k$ in $[\xi,\zeta]$, so we obtain that $(p_k<0)$
    holds at $\zeta$, as well.
  \item In the second case
    \begin{displaymath}
      \exists x\Bigl(
      p = 0 \land
      \bigwedge_{i=1}^{|\sgnseq{s}|} (p^{(i)}\mathrel{\sigma(s_i)}0) \land
      p_k = 0 \land
      \bigwedge_{i=1}^{|\sgnseq{t}|} (p_k^{(i)}\mathrel{\sigma(t_i)}0)
      \Bigr)
    \end{displaymath}
    holds for some $\sgnseq{t}$, and $\sgnr((p_k,\sgnseq{t})) = -1$. Again, from
    the uniqueness of $\xi\in\R$ satisfying \ref{C1}--\ref{C3} we deduce that
    $p_k=0$ holds at $\xi$. Finally, since the right sign of $(p_k,\sgnseq{t})$
    is negative, and there is no root of $p_k$ in $]\xi,\zeta]$, so we conclude
    that $p_k<0$ holds at $\zeta$.
  \end{enumerate}
  If $\varrho_k$ is ``$>$,'' then the proof is done similarly as for ``$<$.'' If
  $\varrho_k$ is ``$\neq$,'' we apply the lemma either to $(p_k >
  0)[\msub{x}{(p,\sgnseq{s}})+\varepsilon]$ or to $(p_k <
  0)[\msub{x}{(p,\sgnseq{s}})+\varepsilon]$ to obtain that $(p_k\neq 0)$ holds
  for $\zeta$. When $\varrho\in\{\leqslant,\geqslant\}$ the proof is similar.

  To finish the proof of (ii), we assume that for every $l\in K$ the formulas
  $(p_l<0)[\msub{x}{(p,\sgnseq{s})+\varepsilon}]$ and
  $(p_l>0)[\msub{x}{(p,\sgnseq{s})+\varepsilon}]$ do not hold. Let $\zeta$ be
  any real number. Let now $k\in K$, and assume that
  $(p_k\mathrel{\varrho_k}0)[\msub{x}{(p,\sgnseq{s})+\varepsilon}]$ holds. Our
  assumption together with the definition of
  $[\msub{x}{(p,\sgnseq{s})+\varepsilon}]$ imply that
  $(p_k\mathrel{\varrho_k}0)[\msub{x}{(p,\sgnseq{s})+\varepsilon}]$ is
  equivalent to $(p_k=0)[\msub{x}{(p,\sgnseq{s})+\varepsilon}]$, which shows
  that $(p_k=0)$ holds for $\zeta$.
\end{proof}

\begin{lemma}\label{LE:semantics2}
  Let $(q\mathrel{\varrho}0)$ be an atomic formula, where $q\in\Zym[x]$, and
  $\varrho\in\{=,\neq,<,\leqslant,\geqslant,>\}$. Consider a parametric root
  $(p,\sgnseq{s})$. Assign values $\alpha_1$, \dots,~$\alpha_m\in\R$ to the
  variables $y_1$, \dots,~$y_m$. Assume that $\guard{p}{\sgnseq{s}}$ holds, and
  denote the unique real number satisfying conditions \ref{C1}--\ref{C3} by
  $\xi$. Then the following hold:
  \begin{enumerate}[label=(\roman*)]
  \item If $\xi$ satisfies $(q\mathrel{\varrho}0)$, then
    $(q\mathrel{\varrho}0)[\msub{x}{(p,\sgnseq{s})}]$ holds.
  \item Let $\xi'\in\R$ be such that all $\zeta\in\left]\xi,\xi'\right[$ satisfy
    $(q\mathrel{\varrho}0)$. Then
    $(q\mathrel{\varrho}0)[\msub{x}{(p,\sgnseq{s})+\varepsilon}]$ holds.
  \end{enumerate}
\end{lemma}
\begin{proof}
  Since $(q\mathrel{\varrho}0)$ holds for $\xi\in\R$, and $\xi$ is unique such
  that \ref{C1}--\ref{C3} hold, we obtain that
  \begin{displaymath}
    \Bigl(
    p = 0 \land
    \bigwedge_{i=1}^{|\sgnseq{s}|} (p^{(i)}\mathrel{\sigma(s_i)}0) \land
    (q\mathrel{\varrho}0)
    \Bigr).
  \end{displaymath}
  holds for $\xi$. By the definition of $[\msub{x}{(p,\sgnseq{s})}]$, we see
  that $(q\mathrel{\varrho}0)[\msub{x}{(p,\sgnseq{s})}]$ holds. This proves (i).

  To prove (ii), we begin with the case when $\varrho$ is ``$=$.'' Since $(q=0)$
  holds in a non-empty interval, $q$ is the zero polynomial, so
  $(q=0)[\msub{x}{(p,\sgnseq{s})+\varepsilon}]$ follows. If $\varrho$ is
  ``$\neq$,'' $(q\neq 0)$ holds in a non-empty interval, so $q$ is a non-zero
  polynomial. Therefore, $(q\mathrel\neq
  0)[\msub{x}{(p,\sgnseq{s})+\varepsilon}]$ holds. If $\varrho$ is ``$<$,''
  there are two possibilities. If $(q<0)$ holds at $\xi$, then by (i) we have
  that $(q<0)[\msub{x}{(p,\sgnseq{s})}]$ holds, so
  $(q<0)[\msub{x}{(p,\sgnseq{s}})+\varepsilon]$ holds, as well. If $(q<0)$ does
  not hold for $\xi$, then $(q=0)$ holds for $\xi$, because $(q<0)$ holds for
  all $\zeta\in\left]\xi,\xi'\right[$. Therefore, there is a Thom code
  $\sgnseq{t}$ identifying $\xi$ as a root of $q$, and the right sign of
  $(q,\sgnseq{t})$ is negative. This ensures that
  $(q<0)[\msub{x}{(p,\sgnseq{s}})+\varepsilon]$ holds in this case, as well. The
  proof for the case when $\varrho$ is ``$>$'' is similar. Finally, the cases
  $\varrho$ is ``$\leqslant$'' or ``$\geqslant$,'' boil down to cases which we
  have just proven. Finally, we conclude that (ii) holds.
\end{proof}

\begin{theorem}[Correctness of the Elimination Set]\label{correctness1}
  Consider a formula $\varphi$ of degree at most $n$ in $x$, which is an
  $\land$-$\lor$-combination of atomic formulas
  $\{p_i\mathrel{\varrho_i}0\}_{i\in I}$, where $p_i\in\Zym[x]$, and
  $\varrho_i\in\{=,\neq,<,\leqslant,\geqslant,>\}$. Let $E_i$ be the set of
  elimination terms generated by $(p_i\mathrel{\varrho_{i}}0)$ as described
  above. Then the following is an elimination set for $\exists x\varphi$:
  \begin{displaymath}
    E=\bigcup_{i\in I}E_i \cup \{-\infty\}.
  \end{displaymath}
\end{theorem}
\begin{proof}
  We have to show that
  \begin{displaymath}
    \R\models\exists x\varphi
    \longleftrightarrow
    \bigvee_{e\in E}\varphi[\msub{x}{e}].
  \end{displaymath}
  Assign values $\alpha_1$, \dots,~$\alpha_m\in\R$ to variables $y_{1}$,
  \dots,~$y_m$. The sets $S_i\subseteq\R$ satisfying
  $(p_i\mathrel{\varrho_i}0)$---and therefore also the set $S\subseteq\R$
  satisfying $\varphi$---are now finite unions of pairwise disjoint (closed,
  half-closed, open, semi-infinite, or infinite) intervals. We show that
  $\bigvee_{e\in E}\varphi[\msub{x}{e}]$ holds if and only if $S\neq\emptyset$.
  Since $\alpha_1$, \dots,~$\alpha_m$ are arbitrary real values, this will imply
  the theorem.

  First assume that $\bigvee_{e\in E}\varphi[\msub{x}{e}]$ holds. Thus,
  $\varphi[\msub{x}{e}]$ holds for some $e\in E$. If
  $\varphi[\msub{x}{-\infty}]$ holds, then the set of values satisfying
  $\varphi$ is unbounded from below, i.e., non-empty.

  Now suppose that $\varphi[\msub{x}{(p,\sgnseq{s})}]$ holds for some parametric
  root $(p,\sgnseq{s})$. By Lemma~\ref{LE:semantics1}(i), there exists
  $\vartheta\in\R$ such that for any $i\in I$ we have:
  $(p_i\mathrel{\varrho_i}0)$ holds for $\vartheta$ whenever
  $(p_i\mathrel{\varrho_i}0)[\msub{x}{(p,\sgnseq{s})}]$ holds. The fact that
  $\varphi$ is an $\land$-$\lor$-combination of atomic formulas ensures that
  $\varphi$ holds for $\vartheta$.

  Assume that $\varphi[\msub{x}{(p,\sgnseq{s})+\varepsilon}]$ holds for some
  parametric root $(p,\sgnseq{s})$. By Lemma~\ref{LE:semantics1}(ii), there
  exists $\zeta\in\R$ such that for any $i\in I$ we have:
  $(p_i\mathrel{\varrho_i}0)$ holds for $\zeta$ whenever
  $(p_i\mathrel{\varrho_i}0)[\msub{x}{(p,\sgnseq{s})}]$ holds. Again, $\varphi$
  is an $\land$-$\lor$-combination of atomic formulas, so $\varphi$ holds for
  $\zeta$.

  We continue by assuming that $S$ is non-empty. To begin with, notice that if
  $S$ is unbounded from below, then $\varphi[\msub{x}{-\infty}]$ holds. In the
  following we therefore assume that $S$ is bounded from below. Since $S$ is a
  finite union of pairwise disjoint (closed, half-closed, open, semi-infinite,
  or infinite) intervals, there exists infimum $\xi$ of $S$ such that $\xi$ is a
  root of some polynomial occurring in $\varphi$. There are two cases to
  consider.

  In the first case $\xi\in S$. This implies that there is some $i\in I$ such
  that $(p_i\mathrel{\varrho_i}0)$ holds for $\xi$, and
  $\varrho_i\in\{=,\leqslant,\geqslant\}$. Moreover, $\xi$ is a root of $p_i$
  with Thom code $\sgnseq{r}$. By the definition of $E_i$ we conclude that
  $(p_i,\sgnseq{r})\in E$. Now Lemma~\ref{LE:semantics2}(i) together with the
  fact that $\varphi$ is an $\land$-$\lor$ combination of atomic formulas ensure
  that $\varphi[\msub{x}{(p_i,\sgnseq{r})}]$ holds.

  In the second case $\xi\notin S$, so there is some $i\in I$ such that
  $(p_i\mathrel{\varrho_i}0)$, $\varrho_i\in\{{\neq,}<,>\}$, holds for all
  $\zeta\in\left]\xi,\xi'\right[$, where
  \begin{displaymath}
    \xi' = \min\{\delta\mid j\in I\land\deg p_j>0 \land p_j(\delta) = 0 \land
    \delta > \xi\}.
  \end{displaymath}
  Furthermore, $\xi$ is a root of $p_i$ with Thom code $\sgnseq{r}$. By the
  definition of $E_i$ we have $(p_i,\sgnseq{r})+\varepsilon\in E_i$. Now
  Lemma~\ref{LE:semantics2}(ii) together with the fact that $\varphi$ is an
  $\land$-$\lor$ combination of atomic formulas ensure that
  $\varphi[\msub{x}{(p_i,\sgnseq{r})+\varepsilon}]$ holds. This concludes the
  proof of the theorem.
\end{proof}

We conclude this section with Algorithm~\ref{ALG:elimination1}, which eliminates
a single existential quantifier. We explicitly point to places where the
external quantifier elimination algorithm $\qealg$ is used. The correctness of
Algorithm~\ref{ALG:elimination1} follows from Theorem~\ref{correctness1}.
\begin{algorithm}[t]
  \caption{Elimination of $\exists x$}
  \label{ALG:elimination1}
  \DontPrintSemicolon
  \LinesNumbered
  \KwIn{$\varphi$\;
    $\varphi$ is an $\land$-$\lor$-combination of atomic formulas
    $(p_i\mathrel{\varrho_i}0)$, $i\in I$, where
    $\varrho_i\in\{=,\neq,<,\leqslant,\geqslant,>\}$, and $p_i\in\Zym[x]$, with
    $\deg p_i\leqslant n$.\;
  }
  \KwOut{$\psi$\;
    $\psi$ is a quantifier-free formula equivalent to $\exists x\varphi$.\;
  }
  \ForEach{$i\in I$}{
    $E_i := \emptyset$\;
    \ForEach{$\sgnseq{s}\in(\{-1,0,1\}^n\setminus\{(0,\dots,0)\})$}{
      Use $\qealg$ to compute $\guard{p_i}{\sgnseq{s}}$.\;
      \If{$\guard{p_i}{\sgnseq{s}}$ is satisfiable} {
        \If{$\varrho\in\{=,\leqslant,\geqslant\}$}{
          $E_i := E_i \cup \{(p_i,\sgnseq{s})\}$\;
        }
        \Else{
          $E_i := E_i \cup \{(p_i,\sgnseq{s})+\varepsilon\}$\;
        }
      }
    }
  }
  $\psi := \false$\;
  \ForEach{$e\in\left(\{-\infty\}\cup\bigcup_{i\in I}E_i\right)$}{
    Use $\qealg$ to compute $\varphi[\msub{x}{e}]$.\;
    $\psi := \psi\lor\varphi[\msub{x}{e}]$\;
  }
  \Return $\psi$\;
\end{algorithm}

\section{Smaller Elimination Sets}\label{SE:smaller}
In this section we are going to considerably reduce the size of our elimination
sets by generalizing a well-known idea from the linear quantifier
elimination~\cite{LoosWeispfenning:1993a} to arbitrary degrees: An elimination
term $t$ is added to the elimination set for $\varphi$ only if $t$ possibly
represents a \emph{lower} bound of a satisfying interval of $\varphi$ for some
choice of parameters $y_1$, \dots,~$y_m$.

Consider a single atomic formula $(p\leqslant 0)$, where $p\in\Zym[x]$. Let
$\alpha_1$, \dots,~$\alpha_m\in\R$, and assume that the polynomial $p\evama$ is
of positive degree and has a real root $\xi\in\R$. There are four possibilities
how $p\evama$ can look in the neighborhood of $\xi$, all of which are pictured
in Figure~\ref{FIG:situations}. The corresponding satisfying sets of
$(p\leqslant 0)$ are shown in red. Observe that only in the second and the third
case $\xi$ is a lower bound of the respective satisfying interval of
$(p\leqslant 0)$, i.e., only when $p\evama$ is positive on the left-hand side of
$\xi$. Consequently, if $p\evama$ is negative on the left-hand side of $\xi$,
then there is either another root $\zeta$ smaller than $\xi$, or $p\evama$ holds
at $-\infty$.

\begin{figure}[t]
  \centering
  \begin{tikzpicture}[line cap=round,line join=round,>=triangle 45,x=6.0cm,y=6.0cm]
    \draw[color=black] (-0.21,0) -- (0.29,0);
    \foreach \x in {-0.3,-0.25,-0.2,-0.15,-0.1,0.15,0.2,0.25}
    \draw[shift={(\x,0)},color=black] (0pt,-2pt);
    \clip(-0.31,-0.19) rectangle (0.29,0.22);
    \draw[line width=2pt, smooth,samples=100,domain=-0.3081407636834072:0.2911130362696604] plot(\x,{0-(\x)*((\x)-1)});
    \draw [line width=2pt,color=ffqqqq,domain=-0.3081407636834072:0.0] plot(\x,{(-0-0*\x)/-0.01});
    \begin{scriptsize}
      \fill [color=ffqqqq] (0,0) circle (2.5pt);
    \end{scriptsize}
  \end{tikzpicture}
  \qquad
  \begin{tikzpicture}[line cap=round,line join=round,x=6.0cm,y=6.0cm]
    \draw[color=black] (-0.31,0) -- (0.20,0);
    \foreach \x in {-0.3,-0.25,-0.2,-0.15,-0.1,0.15,0.2,0.25}
    \draw[shift={(\x,0)},color=black] (0pt,-2pt);
    \clip(-0.31,-0.19) rectangle (0.29,0.22);
    \draw[line width=2pt, smooth,samples=100,domain=-0.31:0.29111303626966045] plot(\x,{(\x)*((\x)-1)});
    \draw [line width=2pt,color=ffqqqq,domain=0.0:0.29111303626966045] plot(\x,{(-0-0*\x)/0.01});
    \begin{scriptsize}
      \fill [color=ffqqqq] (0,0) circle (2.5pt);
    \end{scriptsize}
  \end{tikzpicture}
  \\[6ex]
  \begin{tikzpicture}[line cap=round,line join=round,>=triangle 45,x=6.0cm,y=6.0cm]
    \draw[color=black] (-0.31,0) -- (0.29,0);
    \foreach \x in {-0.3,-0.25,-0.2,-0.15,-0.1,0.15,0.2,0.25}
    \draw[shift={(\x,0)},color=black] (0pt,-2pt);
    \clip(-0.31,-0.19) rectangle (0.29,0.22);
    \draw[line width=2pt, smooth,samples=100,domain=-0.31:0.29111303626966045] plot(\x,{5*(\x)^2});
    \begin{scriptsize}
      \fill [color=ffqqqq] (0,0) circle (2.5pt);
    \end{scriptsize}
  \end{tikzpicture}
  \qquad
  \begin{tikzpicture}[line cap=round,line join=round,>=triangle 45,x=6.0cm,y=6.0cm]
    \draw[color=black] (-0.21,0) -- (0.20,0);
    \foreach \x in {-0.3,-0.25,-0.2,-0.15,-0.1,0.15,0.2,0.25}
    \draw[shift={(\x,0)},color=black] (0pt,-2pt);
    \clip(-0.31,-0.19) rectangle (0.29,0.22);
    \draw[line width=2pt, smooth,samples=100,domain=-0.31:0.29111303626966045] plot(\x,{0-5*(\x)^2});
    \draw [line width=2pt,color=ffqqqq,domain=-0.31:0.0] plot(\x,{(-0-0*\x)/-0.01});
    \draw [line width=2pt,color=ffqqqq,domain=0.0:0.29111303626966045] plot(\x,{(-0-0*\x)/0.01});
    \begin{scriptsize}
      \fill [color=ffqqqq] (0,0) circle (2.5pt);
    \end{scriptsize}
  \end{tikzpicture}
  \caption{$p\evama$ near its root $\xi$\label{FIG:situations}}
\end{figure}

Similar ideas apply to all other relations, which motivates the following
revised definition of the sets of elimination terms $E_i'$ generated by
$(p_i\mathrel{\varrho_i}0)$, where again $(p_i,\sgnseq{s})$ is a parametric root
of $p_i$:
\begin{align*}
  (p_i = 0) &:\quad \left\{\,r \mid r = (p_i,\overline{s})\,\right\}\\
  (p_i \neq 0) &:\quad \left\{\, r+\varepsilon \mid r =
  (p_i,\overline{s})\,\right\}\\
  (p_i < 0) &:\quad \left\{\, r+\varepsilon \mid r =
  (p_i,\overline{s})\land\sgnr(r)=-1\,\right\}\\
  (p_i > 0) &:\quad \left\{\, r+\varepsilon \mid r =
  (p_i,\overline{s})\land\sgnr(r)=1\,\right\}\\
  (p_i \leqslant 0) &:\quad \left\{\, r \mid r =
  (p_i,\overline{s})\land\sgnl(r)=1\,\right\}\\
  (p_i \geqslant 0) &:\quad \left\{\, r \mid r =
  (p_i,\overline{s})\land\sgnl(r)=-1\,\right\}.
\end{align*}

\begin{theorem}[Correctness of the Smaller Elimination Set]\label{THM:correctness2}
  Theorem~\ref{correctness1} remains correct for $E_i'$ instead of $E_i$.
\end{theorem}
\begin{proof}
  We have to show that
  \begin{displaymath}
    \R\models\exists x\varphi
    \longleftrightarrow
    \bigvee_{e\in E}\varphi[\msub{x}{e}].
  \end{displaymath}
  We proceed similarly as in the proof of Theorem~\ref{correctness1}. Assign
  values $\alpha_1$, \dots,~$\alpha_m\in\R$ to variables $y_{1}$, \dots,~$y_m$,
  respectively. Again, the sets $S_i\subseteq\R$ satisfying
  $(p_i\mathrel{\varrho_i}0)$ and also the set $S\subseteq\R$ satisfying
  $\varphi$ are finite unions of pairwise disjoint (closed, half-closed, open,
  semi-infinite, or infinite) intervals. We show that $\bigvee_{e\in
    E}\varphi[\msub{x}{e}]$ holds if and only if $S\neq\emptyset$. This will
  imply the theorem.

  If $\bigvee_{e\in E}\varphi[\msub{x}{e}]$ holds, the proof of the fact that
  $S\neq\emptyset$ is exactly the same as in the proof of
  Theorem~\ref{correctness1}.

  Assume that $S$ is non-empty. To begin with, notice that if $S$ is unbounded
  from below, then $\varphi[\msub{x}{-\infty}]$ holds. In the following we
  therefore assume that $S$ is bounded from below. Since $S$ is a finite union
  of pairwise disjoint (closed, half-closed, open, semi-infinite, or infinite)
  intervals, there exists infimum $\xi$ of $S$ such that $\xi$ is a root of some
  polynomial occurring in $\varphi$. There are two cases to consider.

  In the first case $\xi\in S$. Define the following sets of indices:
  \begin{eqnarray*}
    I_1 & = & \{k\in I\mid\text{$\varrho_k$ is ``$=$''}\land \deg p_k>0\land p_k(\xi) = 0\},\\
    I_2 & = & \{k\in I\mid\text{$\varrho_k$ is ``$\leqslant$''}\land \deg p_k>0\land p_k(\xi) = 0\},\\
    I_3 & = & \{k\in I\mid\text{$\varrho_k$ is ``$\geqslant$''}\land \deg p_k>0\land p_k(\xi) = 0\}.
  \end{eqnarray*}
  Since $\xi\in S$, we have $(I_1\cup I_2\cup I_3)\neq\emptyset$.

  If $I_1\neq\emptyset$, then there exists $i\in I$ such that $\xi$ is a root of
  $p_i$ with Thom code $\sgnseq{s}$. By the definition of $E_i'$, it follows
  that $(p_i,\sgnseq{r})\in E$. Now Lemma~\ref{LE:semantics2}(i) together with
  the fact that $\varphi$ is an $\land$-$\lor$ combination of atomic formulas
  ensure that $\varphi[\msub{x}{(p_i,\sgnseq{s})}]$ holds.

  Now assume that $I_1=\emptyset$. We show that there is either some $i\in I_2$
  such that $p_i$ is positive at $\xi-\varepsilon$, or some $j\in I_3$ such that
  $p_j$ is negative at $\xi-\varepsilon$. Assume the opposite, i.e., for every
  $k\in I_2\cup I_3$, the atomic formula $(p_k\mathrel{\varrho_k}0)$ holds at
  $\xi-\varepsilon$. Let $\zeta\in\left]\xi',\xi\right[$, where
  \begin{displaymath}
    \xi' = \max\{\delta\mid l\in I\land\deg p_l>0 \land p_l(\delta) = 0 \land
    \delta < \xi\}.
  \end{displaymath}
  Since $I_1=\emptyset$, the following implication holds for all $m\in I$: If
  $(p_m\mathrel{\varrho}0)$ holds for $\xi$, then it holds for $\zeta$. This
  together with the fact that $\varphi$ is an $\land$-$\lor$-combination of
  atoms implies that $\varphi$ holds for $\zeta$; a contradiction. Therefore,
  there is either $i\in I_2$ such that $p_i$ is positive for $\xi-\varepsilon$,
  or $j\in I_3$ such that $p_j$ is negative for $\xi-\varepsilon$. In the first
  case $\xi$ is a root of $p_i$ with Thom code $\sgnseq{s}$, so the definition
  of $E_i'$ ensures that $(p_i,\sgnseq{s})\in E_i'$. Now,
  Lemma~\ref{LE:semantics2}(i) together with the fact that $\varphi$ is an
  $\land$-$\lor$ combination of atoms ensure that
  $\varphi[\msub{x}{(p_i,\sgnseq{s})}]$ holds. Finally, in the second case $\xi$
  is a root of $p_j$ with Thom code $\sgnseq{r}$, so the definition of $E_j$
  ensures that $(p_j,\sgnseq{r})\in E_j$. Again, Lemma~\ref{LE:semantics2}(i)
  together with the fact that $\varphi$ is an $\land$-$\lor$-combination of
  atoms ensure that $\varphi[\msub{x}{(p_j,\sgnseq{r})}]$ holds.

  We continue by assuming that $\xi\notin S$. Define the following sets of
  indices:
  \begin{eqnarray*}
    I_4 & = & \{k\in I\mid\text{$\varrho_k$ is ``$\neq$''}\land \deg p_k>0\land p_k(\xi) = 0\},\\
    I_5 & = & \{k\in I\mid\text{$\varrho_k$ is ``$<$''}\land \deg p_k>0\land p_k(\xi) = 0\},\\
    I_6 & = & \{k\in I\mid\text{$\varrho_k$ is ``$>$''}\land \deg p_k>0\land p_k(\xi) = 0\}.
  \end{eqnarray*}
  Since $\xi\notin S$ it follows that $(I_4\cup I_5\cup I_6)\neq\emptyset$.
  Define
  \begin{displaymath}
    \xi' = \min\{\delta\mid k\in I\land\deg p_k>0 \land p_k(\delta) = 0 \land
    \delta > \xi\}.
  \end{displaymath}

  If $I_4\neq\emptyset$, then there exists $i\in I$ such that
  $(p_i\mathrel{\varrho_i}0)$ holds for all $\zeta\in\left]\xi,\xi'\right[$.
  Moreover, $\xi$ is a root of $p_i$ with Thom code $\sgnseq{s}$. By the
  definition of $E_i'$, it follows that $(p_i,\sgnseq{s})+\varepsilon\in E$. Now
  Lemma~\ref{LE:semantics2}(ii) together with the fact that $\varphi$ is an
  $\land$-$\lor$ combination of atomic formulas ensure that
  $\varphi[\msub{x}{(p_i,\sgnseq{s})}]$ holds.

  Now assume that $I_4=\emptyset$. We show that there is either some $i\in I_5$
  such that $p_i$ is negative at $\xi+\varepsilon$ or some $j\in I_6$ such that
  $p_j$ is positive at $\xi+\varepsilon$. Assume the opposite, i.e., for every
  $k\in I_5\cup I_6$, the atomic formula $(p_k\mathrel{\varrho_k}0)$ does not
  hold at $\xi+\varepsilon$. Since $I_4=\emptyset$, the following implication
  holds for all $m\in I$: If $(p_m\mathrel{\varrho}0)$ holds in
  $\left]\xi,\xi'\right[$, then it holds for $\xi$. This together with the fact
  that $\varphi$ is an $\land$-$\lor$-combination of atomic formulas ensure that
  $\varphi$ holds for $\xi$; a contradiction. Therefore, there is either $i\in
  I_5$ such that $p_i$ is negative at $\xi+\varepsilon$ or some $j\in I_6$ such
  that $p_j$ is positive at $\xi+\varepsilon$. In the first case $\xi$ is a root
  of $p_i$ with Thom code $\sgnseq{s}$, so the definition of $E_i'$ ensures that
  $(p_i,\sgnseq{s})+\varepsilon\in E_i'$. Now Lemma~\ref{LE:semantics2}(ii)
  together with the fact that $\varphi$ is an $\land$-$\lor$-combination of
  atoms guarantee that $\varphi[\msub{x}{(p_i,\sgnseq{s})+\varepsilon}]$ holds.
  In the second case $\xi$ is a root of $p_j$ with Thom code $\sgnseq{r}$, so
  the definition of $E_j$ ensures that $(p_j,\sgnseq{r})+\varepsilon\in E_j$.
  Again, Lemma~\ref{LE:semantics2}(ii) together with the fact that $\varphi$ is
  an $\land$-$\lor$-combination of atoms ensure that
  $\varphi[\msub{x}{(p_j,\sgnseq{r})}+\varepsilon]$ holds. This finishes the
  proof of the theorem.
\end{proof}

Theorem~\ref{THM:correctness2} ensures the correctness of
Algorithm~\ref{ALG:elimination2}, which employs the idea of smaller elimination
sets.
\begin{algorithm}[t]
  \caption{Elimination of $\exists x$ (smaller elimination sets)}
  \label{ALG:elimination2}
  \DontPrintSemicolon
  \LinesNumbered
  \KwIn{$\varphi$\;
    $\varphi$ is an $\land$-$\lor$-combination of atomic formulas
    $(p_i\mathrel{\varrho_i}0)$, $i\in I$, where
    $\varrho_i\in\{=,\neq,<,\leqslant,\geqslant,>\}$, and $p_i\in\Zym[x]$, with
    $\deg p_i\leqslant n$.\;
  }
  \KwOut{$\psi$\;
    $\psi$ is a quantifier-free formula equivalent to $\exists x\varphi$.\;
  }
  \ForEach{$i\in I$}{
    $E_i' := \emptyset$\;
    \ForEach{$\sgnseq{s}\in(\{-1,0,1\}^n\setminus\{(0,\dots,0)\})$}{
      Use $\qealg$ to compute $\guard{p_i}{\sgnseq{s}}$.\;
      \If{$\guard{p_i}{\sgnseq{s}}$ is satisfiable} {
        \If{$\varrho_i$ is ``$=$''}{
          $E_i' := E_i' \cup \{(p_i,\sgnseq{s})\}$\;
        }
        \If{$\varrho_i$ is ``$\neq$''}{
          $E_i' := E_i' \cup \{(p_i,\sgnseq{s})+\varepsilon\}$\;
        }
        \If{$\varrho_i$ is ``$<$'' and $\sgnr((p_i,\sgnseq{s}))=-1$ }{
          $E_i' := E_i' \cup \{(p_i,\sgnseq{s})+\varepsilon\}$\;
        }
        \If{$\varrho_i$ is ``$>$'' and $\sgnr((p_i,\sgnseq{s}))=1$ }{
          $E_i' := E_i' \cup \{(p_i,\sgnseq{s})+\varepsilon\}$\;
        }
        \If{$\varrho_i$ is ``$\leqslant$'' and $\sgnl((p_i,\sgnseq{s}))=1$ }{
          $E_i' := E_i' \cup \{(p_i,\sgnseq{s})\}$\;
        }
        \If{$\varrho_i$ is ``$\geqslant$'' and $\sgnl((p_i,\sgnseq{s}))=-1$ }{
          $E_i' := E_i' \cup \{(p_i,\sgnseq{s})\}$\;
        }
      }
    }
  }
  $\psi := \false$\;
  \ForEach{$e\in\left(\{-\infty\}\cup\bigcup_{i\in I}E_i'\right)$}{
    Use $\qealg$ to compute $\varphi[\msub{x}{e}]$.\;
    $\psi := \psi\lor\varphi[\msub{x}{e}]$\;
  }
  \Return $\psi$\;
\end{algorithm}

\section{Relation to Other Work}\label{SE:relation}
The idea to use Thom's lemma as a basis for a virtual substitution-based
quantifier elimination method originally appeared as an outlook in
Weispfenning's article on the quadratic case~\cite{Weispfenning:1997a}.
Weispfenning, too, used an external quantifier elimination algorithm $\qealg$
for realizing the virtual substitution of elimination terms. Besides being way
more explicit, there are some principal differences and novelties in our
framework developed here.

The most important difference is the length of the sign sequences $\sgnseq{s}$
in parametric roots $(p,\sgnseq{s})$, where $p\in\Zym[x]$ is of degree $n$.
Weispfenning's sequences are of length $n-1$, while our framework uses sequences
of length $n$. This implies that $\sgnseq{s}$ imposes stronger restrictions on
the graph of the univariate polynomial $p\evama$ after fixing parameters around
the corresponding root, which we take advantage of to a considerable extent. For
instance, Lemma~\ref{LE:sgnclose} would not hold when using sequences of length
$n-1$. It follows that our notions of the left and the right sign make sense
only when using sequences of length $n$, which is crucial for the correctness of
our optimized elimination sets described in Section~\ref{SE:smaller}.

Weispfenning partitions his sign sequences into \emph{maximally consistent}
sets. This requires the computation of a case distinction on the number of real
roots of $p$. Our formulation of the Theorem~\ref{correctness1} and
Theorem~\ref{THM:correctness2} clearly exhibits that such a case distinction is
not necessary. The deeper reason for this is our following insight: Given
$n\in\N\setminus\{0\}$ there is $p\in\Zym[x]$, where $\deg p = n < m$ such that
all $3^n-1$ sign sequences (excluding the zero sequence) are consistent with
$p$, and therefore need be included in the elimination sets.

Independently, Weispfenning very explicitly discussed the cubic case
in~\cite{Weispfenning:1994a}. This treatment does not rely on an external
algorithm $\qealg$. However, it also does not use Thom codes at all, but relies
on a thorough analysis and case distinction on the \emph{real type} of the
relevant, at most cubic, polynomials. As a matter of fact, many ideas from that
work could be combined with our framework introduced here, either directly or
for realizing $\qealg$.

Our framework is compatible with various generalizations of real quantifier
elimination by virtual substitution. Positive quantifier elimination focuses on
the special case where all variables are known to be strictly
positive~\cite{SturmWeber:08a,SturmWeber:09a}. It is straightforward to adjust
our framework so that this assumption is taken into account during the
construction of elimination sets, decreasing their size. Furthermore, the
external algorithm $\qealg$ could take advantage of this assumption, possibly
constructing shorter quantifier-free formulas.

Extended quantifier elimination keeps track of the elimination terms used and
generates parametric sample solutions for $x$ guarded by quantifier-free
conditions~\cite{Weispfenning:94b,Weispfenning:97d}. The disjunction over these
conditions is actually the regular quantifier elimination result. For details we
refer the reader to our recent work~\cite{KostaSturm:15a-hr}. There we show for
fixed choices of parameters how to use postprocessing to eliminate nonstandard
symbols, like $-\infty$ and $\varepsilon$, from the sample solutions. Both
extended quantifier elimination and our postprocessing are compatible with our
proposed framework.

Generic quantifier elimination makes global assumptions on non-vanishing of
parametric expressions during the elimination process. The result is correct
only under these assumptions, which form part of the output. This saves case
distinctions, which considerably increases efficiency and reduces the overall
size of the output~\cite{DSW:98,Sturm:99a}. In our framework we could assume
that certain derivatives do not vanish, which would further decrease the size of
the elimination sets.

\section{Towards Practical Computations}\label{SE:examples}
Our theoretical discussion in Section~\ref{SE:framework} and
Section~\ref{SE:smaller} uses the external algorithm $\qealg$ \emph{online}
during the computation of $\guard{p}{\sgnseq{s}}$, as well as during virtual
substitution. However, it is a key idea of our framework to use in practice an
\emph{offline} approach. This is based on the following observation, which is
not hard to see:

\begin{proposition}[Finiteness Property]
  Let $n\in\N\setminus\{0\}$. Then our framework requires only finitely many
  formulas to be computed by $\qealg$ to realize quantifier elimination for
  \emph{all} formulas of degree at most $n$.\qed
\end{proposition}

These finitely many formulas are essentially $\guard{p}{\sgnseq{s}}$,
$(q\mathrel{\varrho}0)[\msub{x}{(p,\sgnseq{s})}]$, and
$\nu((p,\sgnseq{s}),(r,\sgnseq{t}))$ for generic $p=u_nx^n+\dots+u_0$,
$q=v_nx^n+\dots+v_0$, and for all possible choices of $\varrho$, $\sgnseq{s}$,
and $\sgnseq{t}$. Then, calls to $\qealg$ in our theoretical description can be
replaced by substitutions into the $u_i$ and $v_j$.

Practical experiments have shown that formulas computed this way using, e.g.,
Qepcad~\cite{Brown:03a} allow to generate such formulas for $n=2$. However,
using these results our framework cannot compete with the classical approach
from~\cite{Weispfenning:1997a}. Having computed part of the formulas for $n=3$,
we assume that Qepcad~\cite{Brown:03a} is going to fail for efficiency reasons
at $n=4$ latest. Our vision is that researchers would create suitable sets of
formulas for given degrees $n$ combining automated methods with human
intelligence. This can even lead to formulas of optimal size, as Lazard has
demonstrated with his famous result on the quartic problem~\cite{Lazard:1988}.
This way, independent research results in symbolic computation could push the
limits of practically applicable quantifier elimination by virtual substitution
towards increasingly higher degrees $n$.

\section*{Acknowledgments}
This research was supported in part by the German Transregional Collaborative
Research Center SFB/TR 14 AVACS and by the ANR/DFG Programme Blanc Project STU
483/2-1 SMArT.

\end{document}